\documentclass[reqno]{amsart}

\usepackage{hyperref,tikz,mathtools,mathptmx}
\usepackage[margin=1in]{geometry}

\newtheorem{thm}{Theorem}
\newtheorem{lem}[thm]{Lemma}
\newtheorem{prop}[thm]{Proposition}
\newtheorem{cor}[thm]{Corollary}

\newcommand{\defeq}{\stackrel{\mathrm{def}}=}
\newcommand{\ccP}{\mathsf{P}}
\newcommand{\BQP}{\mathsf{BQP}}
\newcommand{\shP}{\mathsf{\#P}}
\newcommand{\NP}{\mathsf{NP}}
\newcommand{\poly}{\operatorname{poly}}
\newcommand{\cC}{\mathcal{C}}

\begin{document}
\title{Topological quantum computation is hyperbolic}
\author{Eric Samperton}
\address{Departments of Mathematics and Computer Science, 150 N University St, Purdue University, West Lafayette, Indiana, 47907}
\thanks{Supported in part by NSF grant DMS \#2038020. \url{eric@purdue.edu}}

\date{May 5, 2023}
\begin{abstract}
We show that a topological quantum computer based on the evaluation of a Witten-Reshetikhin-Turaev TQFT invariant of knots can always be arranged so that the knot diagrams with which one computes are diagrams of hyperbolic knots.  The diagrams can even be arranged to have additional nice properties, such as being alternating with minimal crossing number.  Moreover, the reduction is polynomially uniform in the self-braiding exponent of the coloring object.  Various complexity-theoretic hardness results regarding the calculation of quantum invariants of knots follow as corollaries.  In particular, we argue that the hyperbolic geometry of knots is unlikely to be useful for topological quantum computation.
\end{abstract}
\maketitle

\section{Introduction}
\label{s:intro}
Topology, quantum computing, and condensed matter theory have had significant interactions in recent decades, one of the most notable being topological quantum computation.  Freedman, Kitaev, Larsen and Wang established this quantum computing paradigm by showing that there exist unitary topological quantum field theories (TQFTs) based on evaluations of the Jones polynomial at roots of unity for which certain approximations of knot and link invariants are equivalent in power to the usual quantum circuit model of $\BQP$ \cite{FreedmanLarsenWang:universal,FreedmanLarsenWang:two,FreedmanKitaevWang:simulation}.  On the other hand, topology and geometry---especially \emph{hyperbolic} geometry---also have significant interactions in low dimensions, thanks to pioneering work of Thurston and many others.  In particular, there are various conjectural relationships expected to connect the hyperbolic geometry of a knot with its TQFT invariants, the volume conjecture of Kashaev, Murakami and Murakami being perhaps the most infamous \cite{Kashaev:dilogarithm,MurakamiMurakami:volume}.  Moreover, in a slightly different context, Gromov hyperbolicity is known to have substantial algorithmic implications in the theory of finitely presented groups.

Idly combining these various lines of thought leads to a question: Can the hyperbolic geometry of knots be used as a resource to speed-up topological quantum computations?

\subsection{Main results}
\label{ss:main}
We argue \emph{no}.  To this end, we first show that any topological quantum computation can be arranged so that all of the knots one uses to perform quantum calculations are guaranteed to be hyperbolic knots, a fact we might glibly call ``hyperbolic quantum computation."\footnote{To be clear, we have in mind a strict notion of ``topological quantum computation" along the lines of the original papers \cite{FreedmanLarsenWang:universal,FreedmanLarsenWang:two,FreedmanKitaevWang:simulation} that means ``approximate a TQFT invariant of a knot or link inside a closed 3-manifold."  There are now more wide-ranging ideas of topological quantum computation that allow for invariants of colored, trivalent ribbon graphs (meaning projective measurements are performed during the computation), adaptive topological charge measurements (in which the amplitudes with which one computes are topologically protected but are not exactly topological invariants of 3-dimensional objects), or the braiding of Majorana zero modes.  We shall make no attempt to address these broader paradigms.}
In mathematical terms, we show that for any fixed Witten-Reshetikhin-Turaev TQFT invariant of ribbon knots $Z$, there is a polynomial time algorithm that replaces a given knot diagram $K$ with a particularly nice kind of diagram of a hyperbolic knot $K'$ so that $Z(K)$ and $Z(K')$ are essentially equal.  In fact, we have in mind two different ways in which a diagram can be ``particularly nice," and our first two results make each of them precise.

\begin{thm}
Fix a modular fusion category $\cC$ and an object $V$ with a scalar twist $\theta_V$. Let $Z$ be the Witten-Reshetikhin-Turaev invariant of oriented ribbon knots colored by $V$.  Then there exists a classical polynomial time algorithm that converts a given oriented ribbon knot diagram $K$ into a diagram $K'$ such that
\[ Z(K') = \theta_V^{r(K)}Z(K)\]
where $r(K)$ is a polynomial-time computable integer, and $K'$ has the follow additional properties:
\begin{itemize}
\item $K'$ is alternating,
\item $K'$ has minimal crossing number (over all diagrams of equivalent knots), and
\item $K'$ is a hyperbolic knot.
\end{itemize}
If we also allow $\cC$ and $V$ to vary, then the reduction runs in time that is jointly polynomial in the self-braiding exponent $e(V,V)$ and the crossing number of $K$.
\label{thm:reduction1}
\end{thm}

\begin{thm}
Fix a modular fusion category $\cC$ and an object $V$ with a scalar twist $\theta_V$.  Let $Z$ be the Witten-Reshetikhin-Turaev invariant of oriented ribbon knots colored by $V$.  Then there exists a classical polynomial time algorithm that converts a given oriented ribbon knot diagram $K$ into a diagram $K'$ such that $Z(K') = Z(K)$, and $K'$ has the following additional properties:
\begin{itemize}
\item $K'$ is in bridge position,
\item the distance of the induced bridge sphere of $K'$ is known,
\item the bridge number $b(K')$ of $K'$ is minimal (over all diagrams of equivalent knots) and the induced bridge sphere of $K'$ is the unique minimal bridge sphere (up to isotopy), and
\item $K'$ is a hyperbolic knot.
\end{itemize}
If we also allow $\cC$ and $V$ to vary, then the reduction runs in time that is jointly polynomial in the self-braiding exponent $e(V,V)$ and the crossing number of $K$.
\label{thm:reduction2}
\end{thm}

Several minor comments are in order.  First, an object of $\cC$ has a scalar twist exactly when it is a direct sum of simple objects whose twists are equal to each other; in particular, simple objects themselves have scalar twists.  Second, the \emph{self-braiding exponent} $e(V,V)$ of $V$ is the order of the squared braiding isomorphism $br_{V,V}^2 = br_{V,V} \circ br_{V,V}$; see Section \ref{ss:vafa} for more details.  Third, for a discussion of why it is helpful to know either a minimal crossing number or bridge number diagram of a knot, see the last paragraph of Section \ref{ss:takeaways}.  Fourth, we conjecture that there exists a single reduction that builds a knot $K'$ with a combination of \emph{all} of the properties in both Theorems \ref{thm:reduction1} and \ref{thm:reduction2}.  Finally, we also expect that there is an analog of Theorem \ref{thm:reduction2} for closed, triangulated 3-manifolds.

Complexity-theoretic hardness results---both quantum and classical---follow from Theorems \ref{thm:reduction1} and \ref{thm:reduction2} as easy corollaries.  Roughly speaking, each theorem implies that whenever a Witten-Reshetikhin-Turaev invariant $Z$ is known to be hard to compute (or approximate) on knots, then the same hardness result persists even on the restricted class of nice knot diagrams indicated.  We list some concrete applications of this principle here:

\begin{cor}
Fix $i=1$ or $2$.  Let $K$ be a diagram of a knot in $S^3$, treated as computational input, and suppose as a promise that $K$ satisfies all of the listed properties in the conclusion of Theorem $i$.   Then each of the following previously known hardness results continues to hold for such restricted diagrams:
\begin{enumerate}
\item Fix a principal root of unity $t$ that is not of order $1,2,3,4$ or $6$, and let $V(K,t)$ be the value of the Jones polynomial of $K$ at $t$.  Then it is $\BQP$-hard to ``additively approximate" the complex number $V(K,t)$.  More precisely, given a quantum circuit $C$ over some fixed, finite gate set and a fixed error $\epsilon>0$, there is a classical polynomial time reduction that encodes $C$ as a knot diagram $K_{C}$ such that
\begin{equation*}
\left| \mathbb{P}(C \text{ accepts on } \vec{0}) - \left|\frac{V(K_C,t)}{(t^{1/2}+t^{-1/2})^{b(K_C)-1}}\right|^2\right| < \epsilon. \tag{$\star$} \label{e:additive}
\end{equation*}
\item Fix a principal root of unity $t$ that is not of order $1,2,3,4$ or $6$.  Let $0<a<b$ be two positive real numbers, and assume as a further promise that either $|V(K,t)| < a$ or $|V(K,t)|>b$.  Then it is $\shP$-hard, in the sense of Cook-Turing reduction, to decide which inequality holds.
\item Fix a nonabelian finite simple group $G$ and a non-trivial conjugacy class $C \subset G$.  Orient $K$ and let $\gamma \in \pi_1(S^3 - K)$ be a meridian.  Then it is $\shP$-hard, via almost parsimonious reduction, to compute
\[ \#H(K,G,C) \defeq \#\{ \phi: \pi_1(S^3 - K) \to G \mid \phi(\gamma) \in C\}.\]
In particular, deciding when $\pi_1(S^3 - K)$ admits a homomorphism to $G$ with non-cyclic image and $\gamma$ mapping to $C$ is $\NP$-hard.
\end{enumerate}
\label{c:main}
\end{cor}

If we were to remove the listed promises on the diagram $K$, then parts (1), (2), and (3) in Corollary \ref{c:main} are almost verbatim the main results of \cite{FreedmanLarsenWang:two}, \cite{Kuperberg:Jones}, and \cite{me:coloring}, respectively.  The only difference is that \cite{FreedmanLarsenWang:two} proves the $\BQP$-hardness of additively approximating $V(K,t)$ for \emph{links} $K$, while \cite{Kuperberg:Jones} improves their result to \emph{knots} $K$.  The proof of Corollary \ref{c:main} is in Section \ref{ss:cmain}.

We note that Kuperberg conjectured that his results should hold for atoroidal knots \cite[Sec.~5.1]{Kuperberg:Jones}.  Part (2) of Corollary \ref{c:main} confirms this conjecture, because hyperbolic knots are atoroidal.

In particular, parts (1) and (2) of Corollary \ref{c:main} show that when a knot invariant can be used for universal quantum computation, then hyperbolic geometry can not generally be expected to be helpful for achieving significant qualitative improvements in the calculation of such invariants without violating various standard conjectures such as $\BQP \ne \ccP$ or $\ccP^{\shP} \ne \ccP$.  This justifies our claim that hyperbolic geometry is not useful for algorithmic purposes in topological quantum computation.

As explained in their proofs, the three parts of Corollary \ref{c:main} each correspond to a fixed choice of $\cC$ and $V$ in Theorems \ref{thm:reduction1} and \ref{thm:reduction2}.  By applying the uniformity statements in these theorems, we can further expect that whenever we know a \emph{uniform} hardness result for a family of modular fusion categories $\cC_i$ and objects $V_i$, then this uniform result should persist when restricted to the kinds of nice hyperbolic diagrams produced by our main theorems.  Here is a specific manifestation of what we have in mind:

\begin{cor}
Fix $i=1$ or $2$, and fix a polynomial $p(n)$ such that $p(n)>0$ for all integers $n>0$.  Let $K$ be a diagram of a knot in $S^3$, treated as computational input, and suppose as a promise that $K$ satisfies all of the listed properties in the conclusion of Theorem $i$.  Let $n$ be the crossing number of $K$ and let $t = e^{2 \pi i/p(n)}$ be a principle root of unity of order $p(n)$.  Then it is $\BQP$-hard to ``additively approximate" the complex number $V(K,t)$ in the sense of Equation \ref{e:additive} above.
\label{c:uniform}
\end{cor}

Corollary \ref{c:uniform} simply comes from applying Theorems 1 and 2 to the main result of \cite{AharonovArad:uniformity}.  See Section \ref{ss:uniformity} for the proof.

We now turn to three additional take-aways from our results.

\subsection{Additional take-aways}
\label{ss:takeaways}
First, there is a subtle proviso we should make about the phrase ``$K$ is hyperbolic," which stems from the fact that promises do not come with certificates.  Namely, when we say $K$ is promised to be a hyperbolic knot, this does not mean we necessarily have a description of a hyperbolic structure on its complement in hand, say, as an ideal triangulation equipped with an algebraic solution to Thurston's gluing equations.  Rather, the diagram $K$ will satisfy an easily-verifiable combinatorial condition that indirectly guarantees its hyperbolicity by combining work of Johnson and Moriah \cite{JohnsonMoriah:plat} with Thurston's hyperbolization theorem.  So it is still conceivable that a more explicit description of a hyperbolic structure could be exploited, for example, to decide if a hyperbolic knot's fundamental group admits a non-trivial homomorphism to the alternating group $A_5$ in polynomial time.  We consider this unlikely, since if this were the case then by part 3 of Corollary \ref{c:main} either computing hyperbolic structures on hyperbolic knots is $\NP$-hard, or $\ccP=\NP$.  Instead, we conjecture that the diagrams of hyperbolic knots that we construct all admit geometric triangulations constructible in polynomial time.  Recent work of Ham and Purcell on effective Dehn filling lends further evidence to this conjecture \cite{HamPurcell:twisted}.

Second, while we already have argued that hyperbolic geometry is not useful for algorithmic purposes in topological quantum computation, we note that we also have low expectations regarding the possibility of useful applications of ``hyperbolic quantum computation"  to the subject of (quantum) circuit complexity.  One reason is that there exist large quantum circuits that implement the same unitary as the trivial circuit, and our reduction sends these ``complicated-but-trivial" circuits to highly non-trivial knots.  This means the hyperbolic geometry of the encoding knot (more specifically, its volume) can not be used to derive lower bounds on circuit complexity, only upper bounds.

Third, from the perspective of practical algorithmic topology the most important part of Theorem \ref{thm:reduction1} is the promise that $K'$ has minimal crossing number.  Likewise, the most important part of Theorem \ref{thm:reduction2} is the promise that $K''$ has minimal bridge number.   Indeed, a common strategy when computing an invariant of a knot from a diagram $K$ is to first perform some heuristic preprocessing that replaces $K$ with a simpler diagram $K'$ of the same knot.  Depending on the specifics of the implementation, the diagram $K'$ might be simpler because it has smaller crossing number than $K$, or it might be ``thinner" in various senses (e.g. smaller bridge number or tree width).  This strategy is often quite useful in practice, and results on the role of parameterized complexity in knot theory give some sense of the reasons why \cite{BurtonMariaSpreer:TV}.  However, unfortunately, Theorem \ref{thm:reduction1} implies that even with oracle access to a blackbox that replaces an input diagram with an equivalent minimal crossing diagram, computing a TQFT invariant still typically has $\shP$-hard worst case complexity.  Theorem \ref{thm:reduction2} provides a similar conclusion for oracle access to a blackbox that replaces a diagram with an equivalent minimal bridge number diagram.

\subsection{Related results}
In a related vein, our results and their proofs complement the ideas of Cui, Freedman and Wang in \cite{Freedman:axioms,CuiFreedmanWang:axioms}.  These papers begin with the observation that for a TQFT invariant of knots there exist, in addition to the usual Reidemeister moves, more severe diagrammatic simplifications one can make that change the topology of the knot represented by the diagram but nevertheless preserve the value of the invariant.  So one might try to perform a similar simplification strategy as in the previous paragraph using this larger set of moves.  The main result of \cite{CuiFreedmanWang:axioms} builds on the ideas of \cite{Freedman:axioms} to show that if the separation conjecture $\BQP \not\subseteq \mathsf{NDQC1}$ holds, then polynomially many applications of the simplifying moves can not always effect linear simplifications of diagrams.

Finally, we note that the early work \cite{JaegerVertiganWelsh:Jones} of Jaeger, Vertigan and Welsh shows that exact calculation of the value of the Jones polynomial at any $t$ that is not a root of unity of order 1, 2, 3, 4, or 6 is $\shP$-hard for alternating links.  They prove their result by first establishing hardness results for evaluations of the Tutte polynomial of planar graphs.  Then  they identify planar graphs with alternating diagrams of links via the Tait/checkerboard graph construction, and exploit elementary identities between the Tutte polynomial of the Tait graphs and the Jones polynomials of the associated link diagrams.  While their result does not make any promises about crossing number, hyperbolicity, having only one component, \emph{etc.}, we can imagine a strengthening of their result along these lines with a proof that builds on their techniques.  However, we believe our proof is conceptually much simpler, essentially only requiring Vafa's theorem for modular fusion categories and well-known results in combinatorial topology.  In particular, our technique applies to arbitrary Witten-Reshetikhin-Turaev invariants of knots and can be ``plugged into" different kinds of hardness results in a blackbox fashion.

\subsection{Notation caveats}
\label{ss:notation}
We warn the reader that we regularly abuse notation by letting $K$ denote both a knot and a diagram of that knot.  It should always be clear in context which we mean (almost always the latter).  ``Crossing number'' will always mean crossing number of a diagram.  We may also abuse language by saying ``knot" when strictly speaking, in the context of Witten-Reshetikhin-Turaev invariants, we ought to consider \emph{ribbon} knots.  To rectify this, we use the blackboard framing.

\subsection*{Acknowledgements}
We thank Chris Leininger for many helpful conversations about curve complexes.  We also thank Colleen Delaney for very helpful feedback on early drafts of this work, as well as Nathan Dunfield.

\subsection*{Data availability statement}
Data sharing is not applicable to this article as no datasets were generated or analysed during the current study.

\subsection*{Conflict of interest statement}
This work was supported by NSF DMS \#2038020.
The authors have no relevant financial or non-financial interests to disclose.

\section{Knot diagrams and hyperbolic geometry}
\label{s:diagrams}
A knot $K$ in the 3-sphere $S^3$ is \emph{hyperbolic} if there exists a complete hyperbolic structure on its complement $S^3 \setminus K$.  Such a structure can be described explicitly either by finding a triangulation of $S^3 \setminus K$ together with a solution to Thurston's gluing equations, or by finding a discrete and faithful representation $\pi_1(S^3 \setminus K) \to \mathrm{PSL}(2,\mathbb{C})$.  As mentioned in Section \ref{ss:takeaways}, explicit descriptions of hyperbolic structures are not necessary for our purposes.  Thanks to the central role of both knots and hyperbolic geometry throughout 3-manifold topology, there is an extensive literature on the problem of converting combinatorial properties of a knot diagram to hyperbolic geometric properties of the knot.  We make no attempt to summarize this large body of work, but focus on two well-studied such combinatorial properties: alternating diagrams and the \emph{bridge distance} of a plat diagram.  Nothing in this section is new and ribbon framings play no role.

\subsection{Alternating diagrams, nugatory crossings, and Menasco's theorems}
\label{ss:alternating}
A link diagram is called \emph{alternating} when traversing along the diagram results in a sequence of crossings that always alternate between over and under.  Thanks largely to theorems of Menasco \cite{Menasco:alternating} and the resolution of the Tait conjectures \cite{Kauffman:state,Murasugi:conjectures,Murasugi:conjecturesII,Thistlethwaite:tree,Thistlethwaite:Kauffman,MenascoThistlethwaite:flyping}, isotopy of alternating link diagrams is essentially entirely understood.  More precisely, the following appears to be either a folklore result or conjecture: given two $l$-component alternating link diagrams $K$ and $L$ with crossing numbers $m$ and $n$ (respectively), there exists an algorithm to decide if $K$ and $L$ represent isotopic links that runs in time $\poly(l,m,n)$.

As far as we are aware, neither the precise statement of this folklore claim nor its proof have appeared in the literature before, cf.\ \cite{Lackenby:private}.  We do not utilize this full result in the present work, and it would be too much of a digression to attempt to include a full proof.  However, the next two lemmas constitute the easiest, first steps in the algorithm, and we provide some comments at the end of this subsection to indicate how one might complete the algorithm.

A crossing in a knot diagram is called \emph{nugatory} if it looks like
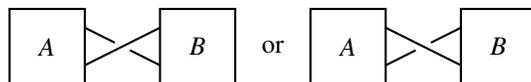
\[ 
\begin{tikzpicture}[thick,scale=.5]
\draw[thick] (0,-1)--(-2,-1) -- (-2,1) -- (0,1) -- (0,-1);
\draw (-1,0) node {$A$};
\draw[thick] (4,-1)--(2,-1) -- (2,1) -- (4,1) -- (4,-1);
\draw (3,0) node {$B$};
\draw (0,.5) -- (2,-.5);
\draw[color=white,fill=white] (1,0) circle[radius=.2];
\draw (0,-.5) -- (2,.5);
\draw (5,0) node {or};
\begin{scope}[xshift=8cm]
\draw[thick] (0,-1)--(-2,-1) -- (-2,1) -- (0,1) -- (0,-1);
\draw (-1,0) node {$A$};
\draw[thick] (4,-1)--(2,-1) -- (2,1) -- (4,1) -- (4,-1);
\draw (3,0) node {$B$};
\draw (0,-.5) -- (2,.5);
\draw[color=white,fill=white] (1,0) circle[radius=.2];
\draw (0,.5) -- (2,-.5);
\end{scope}
\end{tikzpicture}
\]
where $A$ and $B$ are tangle sub-diagrams.  One of the Tait conjectures, independently proved by Kaufman \cite{Kauffman:state}, Murasugi \cite{Murasugi:conjectures} and Thistlethwaite \cite{Thistlethwaite:tree} (all using the Jones polynomial!), says that an alternating link diagram has minimal crossing number whenever it has no nugatory crossings.   A single nugatory crossing can be found and removed in time polynomial in the crossing number of $L$, so removing all of them takes at most polynomial time.  Indeed, if we compute either of the Tait/checkerboard graphs associated to $L$ (which can surely be done in polynomial time), then the nugatory crossings correspond precisely to isolated vertices and length 2 cycles enclosing a single crossing, both of which are easy to identify.  This proves:

\begin{lem}
An alternating link diagram $L$ can be reduced to a minimal crossing alternating diagram in polynomial time.
\label{l:reduced}
\end{lem}

Any knot or link diagram is called \emph{reduced} if all nugatory crossings have been removed as above.  Thus, reduced alternating diagrams are all minimal-crossing.

A link is \emph{prime} if it can not be expressed as a connected sum of two non-trivial links.  A diagram $L$ is \emph{diagrammatically prime} if for each disk $D$ in the plane of the diagram with $\partial D$ meeting $L$ transversely in two non-crossing points, at least one of $D \cap L$ or $(S^3 \setminus D \cap L)$ consists of a single arc with no crossings.  Menasco proved that a non-split, reduced alternating link diagram represents a prime link if and only if it is diagrammatically prime \cite{Menasco:alternating}.  For \emph{reduced} alternating diagrams, non-trivial diagrammatic connected summations correspond exactly to length two cycles in the checkerboard graph.  An innermost such cycle can be identified in polynomial time, from which we may bubble off one prime, reduced summand; by induction on crossing number, we can bubble off all of the connected summands one-by-one in polynomial time.  This proves:

\begin{lem}
A reduced, alternating knot diagram $K$ can be identified as prime in polynomial time, and if $K$ is not prime, then a diagrammatic connected sum decomposition $K = \#_{i=1}^k P_i$ can be identified in polynomial time.  Here each $P_i$ is a non-trivial reduced, prime, alternating knot diagram.
\label{l:prime}
\end{lem}

We note that a version of the lemma can be stated for alternating \emph{links}, not just knots.  However, since we do not need to deal with links later and since it would require some additional care to define connected sums of links, we have decided to avoid this more general statement.

Menasco furthermore proved in \cite{Menasco:alternating} that reduced, prime, alternating diagrams of knots are either trivial, $(2,p)$ torus knots, or hyperbolic.  By the solution to the Tait flyping conjecture \cite{MenascoThistlethwaite:flyping}, a $(2,p)$ torus knot has a unique minimal crossing alternating diagram.\footnote{We note that all of our diagrams should be considered as diagrams on the 2-sphere $S^2$, not as diagrams in $\mathbb{R}^2$.}  Given any diagram, we can check in constant time if it equals (as a diagram) the trivial unknot diagram, and in polynomial time we can check if it equals (again, as a diagram) the standard $(2,p)$ torus knot diagram.  We combine these results in the following:

\begin{lem}
Any reduced, prime, alternating knot diagram is precisely one of the following:
\begin{enumerate}
\item the trivial unknot diagram,
\item the unique minimal crossing alternating diagram of a $(2,p)$ torus knot (in particular, $p$ is odd and $|p|\ge 3$), or
\item a hyperbolic knot.
\end{enumerate}
\label{l:menasco}
Either of the first two types of diagrams can be recognized in polynomial time, and hence so can the third type.
\end{lem}

We conclude this subsection with a brief sketch of an idea for a polynomial time algorithm to solve the isotopy problem for alternating diagrams.  First, we recall yet another theorem of Menasco: if $K$ is an alternating link diagram, then $K$ represents a split link if and only if the diagram is disconnected \cite{Menasco:alternating}.  Now, given $K$ and $L$, we first identify each of their sets of split pieces $K_1,\dots,K_a$, $L_1,\dots,L_b$ (of course if $a\ne b$, then the links are not isotopic).  Lemmas \ref{l:reduced} and \ref{l:prime} together allow us to decompose each of the $K_i$ and $L_j$ into their prime summands in polynomial time.  By uniqueness of prime factorizations, all of the split pieces are determined by their lists of prime summands.  Thus, if we can decide when two prime, reduced, non-split, and alternating link diagrams are isotopic in polynomial time, then we can simply compare these lists to conclude whether or not $K$ and $L$ are isotopic.

Unfortunately, the statement of the flyping conjecture alone is not sufficient to do this is in polynomial time, for the reason that a diagram might have an exponentially large flyping space (note that this would be sufficient to put the isotopy problem in $\NP$).  However, with some additional ideas, we expect that in polynomial time one can compute a normal form for a reduced, alternating, non-split diagram with the property that two diagrams are isotopic if and only if their normal forms are equal.  It is this last step which does not appear in the literature in this precise way, although we conjecture that either the methods of \cite[Sec.~4]{Lackenby:alternating} or \cite{Rankin:alternating2} can be made to work.


\subsection{Plat diagrams and bridge distance}
\label{ss:plat}

For any knot $K$ in $S^3$, a \emph{bridge sphere} $\Sigma$ is an embedded $2$-sphere in $S^3$ such that $K$ intersects both sides of $\Sigma$ in a trivial tangle.  That is, a bridge sphere $\Sigma$ of $K$ cuts $S^3$ into two closed balls $V$ and $W$ such that:
\begin{itemize}
\item $V \cap W = \Sigma$,
\item $K \cap \Sigma$ is $2m$ points for some positive integer $m$,
\item $K \cap V$ can be ambiently isotoped rel $\Sigma$ into $\Sigma$,
\item $K \cap W$ can be ambiently isotoped rel $\Sigma$ into $\Sigma$.
\end{itemize}
Define
\[ \Sigma_K \defeq \Sigma \setminus K, \qquad V_K \defeq V \setminus K, \qquad \text{ and } \ W_K \defeq W \setminus K.\]
The decomposition
\[ S^3 \setminus K = V_K \bigsqcup_{\Sigma_K} W_K\]
is called a \emph{bridge decomposition} of $S^3 \setminus K$, and $m$ is called the \emph{bridge number} of the bridge decomposition.  Bridge decompositions of knots are knot-theoretic analogs of Heegaard splittings of closed, orientable 3-manifolds.

A diagram $K$ in the $xy$-plane is in \emph{bridge position} if, with respect to the $y$ coordinate, all local maxima of the diagram occur above all local minima.  A knot diagram in bridge position encodes an obvious bridge sphere by horizontally slicing through the thickest part.  More precisely, suppose the diagram $K$ lies in the $xy$-plane, considered as a subset of $\mathbb{R}^3$ equipped with the usual coordinates $x,y,z$, and let $S^3$ be $\mathbb{R}^3$ together with a point at infinity $\infty$.  Let $y=c$ be any plane that cuts $K$ so that the minima and maxima of $K$ are on opposite sides.  Then $\Sigma = \{y=c\} \cup \{\infty\}$ is a bridge sphere for $K$.  We call $\Sigma$ an \emph{induced bridge sphere} of the bridge diagram.  The bridge number of $\Sigma$ equals the number of local minima (equivalently, maxima) in the diagram $K$.  If an induced bridge sphere $\Sigma$ of the bridge diagram $K$ has bridge number $m$, then we say $K$ is an \emph{$m$-bridge} diagram.  Every knot-bridge sphere pair can be ambiently isotoped so the bridge sphere is an induced bridge sphere of a bridge diagram.  

Our interest is mainly in bridge diagrams of a specific type.  We say a diagram $K$ is in \emph{$m$-plat position} if:
\begin{itemize}
\item $K$ is in $m$-bridge position,
\item all of the local maxima of $K$ occur above all of the crossings, and
\item all of the local minima of $K$ occur below all of the crossings.
\end{itemize}
The \emph{plats} of the diagram are the arcs where minima or maxima occur.  Every diagram in plat position determines two sets of plats, the \emph{top plats} and the \emph{bottom plats}, each of which is equivalent to a planar matching of $2m$ points.

We shall arrange all of our plat diagram into \emph{rows}, which are tangle subdiagrams of the form
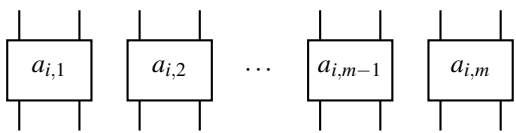
\[
\begin{tikzpicture}[thick,scale=0.8]

   \draw (.8,0) rectangle (2.2,1);
    \draw (1.5,.5) node {$a_{i,1}$};
    \draw (1,0) -- (1,-.5);
    \draw (2,0) -- (2,-.5);
    \draw (1,1) -- (1,1.5);
    \draw (2,1) -- (2,1.5);
    
    \begin{scope}[xshift = 2cm, yshift = 0]
    \draw (.8,0) rectangle (2.2,1);
    \draw (1.5,.5) node {$a_{i,2}$};
    \draw (1,0) -- (1,-.5);
    \draw (2,0) -- (2,-.5);
    \draw (1,1) -- (1,1.5);
    \draw (2,1) -- (2,1.5);
    \end{scope}
    
    \draw (5,.5) node {$\cdots$};
    
    \begin{scope}[xshift = 5cm, yshift = 0]
    \draw (.8,0) rectangle (2.2,1);
    \draw (1.5,.5) node {$a_{i,m-1}$};
    \draw (1,0) -- (1,-.5);
    \draw (2,0) -- (2,-.5);
    \draw (1,1) -- (1,1.5);
    \draw (2,1) -- (2,1.5);
    \end{scope}
    
     \begin{scope}[xshift = 7cm, yshift = 0]
    \draw (.8,0) rectangle (2.2,1);
    \draw (1.5,.5) node {$a_{i,m}$};
    \draw (1,0) -- (1,-.5);
    \draw (2,0) -- (2,-.5);
    \draw (1,1) -- (1,1.5);
    \draw (2,1) -- (2,1.5);
    \end{scope}

\end{tikzpicture}
\]
if $i$ is even, or
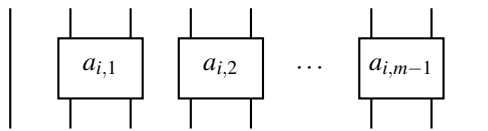
\[
\begin{tikzpicture}[thick,scale=0.8]
   
   \draw (0,-.5) -- (0,1.5);
   \draw (8,-.5) -- (8,1.5);
   
   \draw (.8,0) rectangle (2.2,1);
    \draw (1.5,.5) node {$a_{i,1}$};
    \draw (1,0) -- (1,-.5);
    \draw (2,0) -- (2,-.5);
    \draw (1,1) -- (1,1.5);
    \draw (2,1) -- (2,1.5);
    
    \begin{scope}[xshift = 2cm, yshift = 0]
    \draw (.8,0) rectangle (2.2,1);
    \draw (1.5,.5) node {$a_{i,2}$};
    \draw (1,0) -- (1,-.5);
    \draw (2,0) -- (2,-.5);
    \draw (1,1) -- (1,1.5);
    \draw (2,1) -- (2,1.5);
    \end{scope}
    
    \draw (5,.5) node {$\cdots$};
    
    \begin{scope}[xshift = 5cm, yshift = 0]
    \draw (.8,0) rectangle (2.2,1);
    \draw (1.5,.5) node {$a_{i,m-1}$};
    \draw (1,0) -- (1,-.5);
    \draw (2,0) -- (2,-.5);
    \draw (1,1) -- (1,1.5);
    \draw (2,1) -- (2,1.5);
    \end{scope}
\end{tikzpicture}
\]
if $i$ is odd.  Each box is called a \emph{twist region}, and is labeled by an integer $a_{i,j} \in \mathbb{Z}$ called a \emph{twist coefficient} that indicates a number of half-twists, e.g.
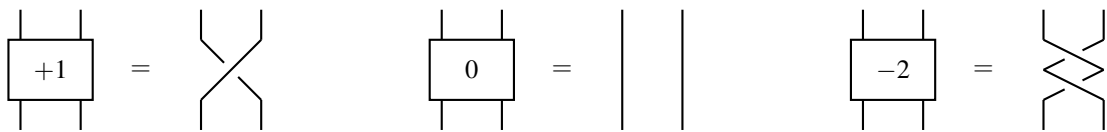
\[
\begin{tikzpicture}[thick,scale=0.8]
\begin{scope}[xshift = -14cm]
\draw (0,0) -- (0,-.5);
\draw (0,1) -- (0,1.5);
\draw (1,0) -- (1,-.5);
\draw (1,1) -- (1,1.5);
\draw (-.2,0) rectangle (1.2,1);
\draw (.5,.5) node {$+1$};
\draw (2,.5) node {$=$};
\begin{scope}[xshift=3cm]
\draw (0,0) -- (0,-.5);
\draw (0,1) -- (0,1.5);
\draw (1,0) -- (1,-.5);
\draw (1,1) -- (1,1.5);
\draw (1,0) -- (0,1);
\draw[draw=white,fill=white] (.5,.5) circle [radius=0.15];
\draw (0,0) -- (1,1);
\end{scope}
\end{scope}

\begin{scope}[xshift = -7cm]
\draw (0,0) -- (0,-.5);
\draw (0,1) -- (0,1.5);
\draw (1,0) -- (1,-.5);
\draw (1,1) -- (1,1.5);
\draw (-.2,0) rectangle (1.2,1);
\draw (.5,.5) node {$0$};
\draw (2,.5) node {$=$};
\begin{scope}[xshift=3cm]
\draw (0,-.5) -- (0,1.5);
\draw (1,-.5) -- (1,1.5);
\end{scope}
\end{scope}

\draw (0,0) -- (0,-.5);
\draw (0,1) -- (0,1.5);
\draw (1,0) -- (1,-.5);
\draw (1,1) -- (1,1.5);
\draw (-.2,0) rectangle (1.2,1);
\draw (.5,.5) node {$-2$};
\draw (2,.5) node {$=$};
\begin{scope}[xshift=3cm]
\draw (0,0) -- (0,-.5);
\draw (0,1) -- (0,1.5);
\draw (1,0) -- (1,-.5);
\draw (1,1) -- (1,1.5);
\draw (0,0) -- (1,.5);
\draw[draw=white,fill=white] (.5,.25) circle [radius=0.15];
\draw (1,0) -- (0,.5);
\begin{scope}[yshift=0.5cm];
\draw (0,0) -- (1,.5);
\draw[draw=white,fill=white] (.5,.25) circle [radius=0.15];
\draw (1,0) -- (0,.5);
\end{scope}
\end{scope}
\end{tikzpicture}
\]
For a typical $m$-plat diagram, many of the twist coefficients $a_{i,j}$ may equal 0.  An $m$-plat diagram with an even number of rows is \emph{standard} if the top plats form this planar matching:
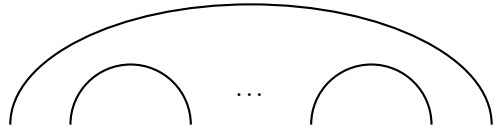
\[
\begin{tikzpicture}[thick,scale=0.8]
\draw (0,0) arc [x radius=4, y radius=2, start angle = 180, end angle = 0];
\draw (1,0) arc [radius=1, start angle = 180, end angle=0];
\draw (4,0.5) node {$\dots$};
\draw (5,0) arc [radius=1, start angle = 180, end angle=0];
\end{tikzpicture}
\]
and the bottom plats form this one:
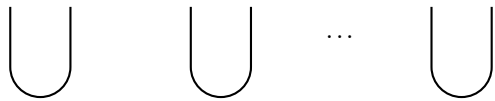
\[
\begin{tikzpicture}[thick,scale=0.8]
\draw (0,0)--(0,-1) arc [radius=.5, start angle = 180, end angle=360] -- (1,0);
\draw (3,0)--(3,-1) arc [radius=.5, start angle = 180, end angle=360] -- (4,0);
\draw (5.5,-.5) node {$\dots$};
\draw (7,0)--(7,-1) arc [radius=.5, start angle = 180, end angle=360] -- (8,0);
\end{tikzpicture}
\]
When an $m$-plat diagram has an odd number of rows, it is \emph{standard} if the bottom plats are as above, but the top plats form this planar matching:
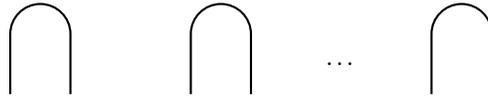
\[
\begin{tikzpicture}[thick,scale=0.8]
\begin{scope}[yscale=-1]
\draw (0,0)--(0,-1) arc [radius=.5, start angle = 180, end angle=360] -- (1,0);
\draw (3,0)--(3,-1) arc [radius=.5, start angle = 180, end angle=360] -- (4,0);
\draw (5.5,-.5) node {$\dots$};
\draw (7,0)--(7,-1) arc [radius=.5, start angle = 180, end angle=360] -- (8,0);
\end{scope}
\end{tikzpicture}
\]
See Figures \ref{f:evenplat} and \ref{f:oddplat}.  Finally, we say an $m$-plat diagram is \emph{highly twisted} if it is standard and $|a_{i,j}| \ge 3$ for all $i$ and $j$.

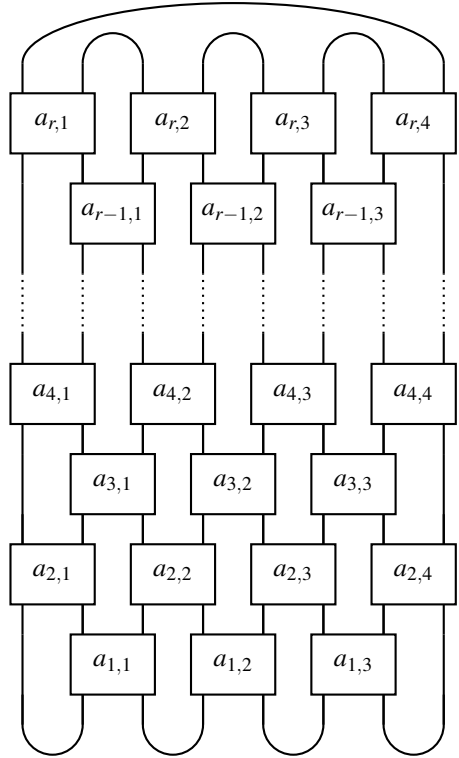
\begin{figure}
\begin{tikzpicture}[thick,scale=0.8]
\draw (0,0) -- (0,-.5) arc [radius = 0.5, start angle = 180, end angle = 360] -- (1,0);
\draw (2,0) -- (2,-.5) arc [radius = 0.5, start angle = 180, end angle = 360] -- (3,0);
\draw (4,0) -- (4,-.5) arc [radius = 0.5, start angle = 180, end angle = 360] -- (5,0);
\draw (6,0) -- (6,-.5) arc [radius = 0.5, start angle = 180, end angle = 360] -- (7,0);

\draw (.8,0) rectangle (2.2,1);
\draw (1.5,.5) node {$a_{1,1}$};
\draw (1,0) -- (1,-.5);
\draw (2,0) -- (2,-.5);
\draw (1,1) -- (1,1.5);
\draw (2,1) -- (2,1.5);

\begin{scope}[xshift = 2cm, yshift = 0]
\draw (.8,0) rectangle (2.2,1);
\draw (1.5,.5) node {$a_{1,2}$};
\draw (1,0) -- (1,-.5);
\draw (2,0) -- (2,-.5);
\draw (1,1) -- (1,1.5);
\draw (2,1) -- (2,1.5);
\end{scope}

\begin{scope}[xshift = 4cm, yshift = 0]
\draw (.8,0) rectangle (2.2,1);
\draw (1.5,.5) node {$a_{1,3}$};
\draw (1,0) -- (1,-.5);
\draw (2,0) -- (2,-.5);
\draw (1,1) -- (1,1.5);
\draw (2,1) -- (2,1.5);
\end{scope}

\begin{scope}[xshift = -1cm, yshift = 1.5cm]
\draw (.8,0) rectangle (2.2,1);
\draw (1.5,.5) node {$a_{2,1}$};
\draw (1,0) -- (1,-.5);
\draw (2,0) -- (2,-.5);
\draw (1,1) -- (1,1.5);
\draw (2,1) -- (2,1.5);
\end{scope}

\begin{scope}[xshift = 1cm, yshift = 1.5cm]
\draw (.8,0) rectangle (2.2,1);
\draw (1.5,.5) node {$a_{2,2}$};
\draw (1,0) -- (1,-.5);
\draw (2,0) -- (2,-.5);
\draw (1,1) -- (1,1.5);
\draw (2,1) -- (2,1.5);
\end{scope}

\begin{scope}[xshift = 3cm, yshift = 1.5cm]
\draw (.8,0) rectangle (2.2,1);
\draw (1.5,.5) node {$a_{2,3}$};
\draw (1,0) -- (1,-.5);
\draw (2,0) -- (2,-.5);
\draw (1,1) -- (1,1.5);
\draw (2,1) -- (2,1.5);
\end{scope}

\begin{scope}[xshift = 5cm, yshift = 1.5cm]
\draw (.8,0) rectangle (2.2,1);
\draw (1.5,.5) node {$a_{2,4}$};
\draw (1,0) -- (1,-.5);
\draw (2,0) -- (2,-.5);
\draw (1,1) -- (1,1.5);
\draw (2,1) -- (2,1.5);
\end{scope}

\draw (0,-.5) -- (0,1);
\draw (7,-.5) -- (7,1);

\begin{scope}[yshift = 3cm]

    \draw (.8,0) rectangle (2.2,1);
    \draw (1.5,.5) node {$a_{3,1}$};
    \draw (1,0) -- (1,-.5);
    \draw (2,0) -- (2,-.5);
    \draw (1,1) -- (1,1.5);
    \draw (2,1) -- (2,1.5);
    
    \begin{scope}[xshift = 2cm, yshift = 0]
    \draw (.8,0) rectangle (2.2,1);
    \draw (1.5,.5) node {$a_{3,2}$};
    \draw (1,0) -- (1,-.5);
    \draw (2,0) -- (2,-.5);
    \draw (1,1) -- (1,1.5);
    \draw (2,1) -- (2,1.5);
    \end{scope}
    
    \begin{scope}[xshift = 4cm, yshift = 0]
    \draw (.8,0) rectangle (2.2,1);
    \draw (1.5,.5) node {$a_{3,3}$};
    \draw (1,0) -- (1,-.5);
    \draw (2,0) -- (2,-.5);
    \draw (1,1) -- (1,1.5);
    \draw (2,1) -- (2,1.5);
    \end{scope}
    
    \begin{scope}[xshift = -1cm, yshift = 1.5cm]
    \draw (.8,0) rectangle (2.2,1);
    \draw (1.5,.5) node {$a_{4,1}$};
    \draw (1,0) -- (1,-.5);
    \draw (2,0) -- (2,-.5);
    \draw (1,1) -- (1,1.5);
    \draw (2,1) -- (2,1.5);
    \end{scope}
    
    \begin{scope}[xshift = 1cm, yshift = 1.5cm]
    \draw (.8,0) rectangle (2.2,1);
    \draw (1.5,.5) node {$a_{4,2}$};
    \draw (1,0) -- (1,-.5);
    \draw (2,0) -- (2,-.5);
    \draw (1,1) -- (1,1.5);
    \draw (2,1) -- (2,1.5);
    \end{scope}
    
    \begin{scope}[xshift = 3cm, yshift = 1.5cm]
    \draw (.8,0) rectangle (2.2,1);
    \draw (1.5,.5) node {$a_{4,3}$};
    \draw (1,0) -- (1,-.5);
    \draw (2,0) -- (2,-.5);
    \draw (1,1) -- (1,1.5);
    \draw (2,1) -- (2,1.5);
    \end{scope}
    
    \begin{scope}[xshift = 5cm, yshift = 1.5cm]
    \draw (.8,0) rectangle (2.2,1);
    \draw (1.5,.5) node {$a_{4,4}$};
    \draw (1,0) -- (1,-.5);
    \draw (2,0) -- (2,-.5);
    \draw (1,1) -- (1,1.5);
    \draw (2,1) -- (2,1.5);
    \end{scope}

    \draw[dotted] (0,3) -- (0,4);
    \draw[dotted] (1,3) -- (1,4);
    \draw[dotted] (2,3) -- (2,4);
    \draw[dotted] (3,3) -- (3,4);
    \draw[dotted] (4,3) -- (4,4);
    \draw[dotted] (5,3) -- (5,4);
    \draw[dotted] (6,3) -- (6,4);
    \draw[dotted] (7,3) -- (7,4);
    
    \draw (0,-.5) -- (0,1);
    \draw (7,-.5) -- (7,1);

\end{scope}

\begin{scope}[yshift = 7.5cm]

    \draw (.8,0) rectangle (2.2,1);
    \draw (1.5,.5) node {$a_{r-1,1}$};
    \draw (1,0) -- (1,-.5);
    \draw (2,0) -- (2,-.5);
    \draw (1,1) -- (1,1.5);
    \draw (2,1) -- (2,1.5);
    
    \begin{scope}[xshift = 2cm, yshift = 0]
    \draw (.8,0) rectangle (2.2,1);
    \draw (1.5,.5) node {$a_{r-1,2}$};
    \draw (1,0) -- (1,-.5);
    \draw (2,0) -- (2,-.5);
    \draw (1,1) -- (1,1.5);
    \draw (2,1) -- (2,1.5);
    \end{scope}
    
    \begin{scope}[xshift = 4cm, yshift = 0]
    \draw (.8,0) rectangle (2.2,1);
    \draw (1.5,.5) node {$a_{r-1,3}$};
    \draw (1,0) -- (1,-.5);
    \draw (2,0) -- (2,-.5);
    \draw (1,1) -- (1,1.5);
    \draw (2,1) -- (2,1.5);
    \end{scope}
    
    \begin{scope}[xshift = -1cm, yshift = 1.5cm]
    \draw (.8,0) rectangle (2.2,1);
    \draw (1.5,.5) node {$a_{r,1}$};
    \draw (1,0) -- (1,-.5);
    \draw (2,0) -- (2,-.5);
    \draw (1,1) -- (1,1.5);
    \draw (2,1) -- (2,1.5);
    \end{scope}
    
    \begin{scope}[xshift = 1cm, yshift = 1.5cm]
    \draw (.8,0) rectangle (2.2,1);
    \draw (1.5,.5) node {$a_{r,2}$};
    \draw (1,0) -- (1,-.5);
    \draw (2,0) -- (2,-.5);
    \draw (1,1) -- (1,1.5);
    \draw (2,1) -- (2,1.5);
    \end{scope}
    
    \begin{scope}[xshift = 3cm, yshift = 1.5cm]
    \draw (.8,0) rectangle (2.2,1);
    \draw (1.5,.5) node {$a_{r,3}$};
    \draw (1,0) -- (1,-.5);
    \draw (2,0) -- (2,-.5);
    \draw (1,1) -- (1,1.5);
    \draw (2,1) -- (2,1.5);
    \end{scope}
    
    \begin{scope}[xshift = 5cm, yshift = 1.5cm]
    \draw (.8,0) rectangle (2.2,1);
    \draw (1.5,.5) node {$a_{r,4}$};
    \draw (1,0) -- (1,-.5);
    \draw (2,0) -- (2,-.5);
    \draw (1,1) -- (1,1.5);
    \draw (2,1) -- (2,1.5);
    \end{scope}

    \draw (0,-.5) -- (0,1);
    \draw (7,-.5) -- (7,1);

\end{scope}

\draw (0,10.5) arc [x radius = 3.5cm, y radius = 1cm, start angle = 180, end angle = 0];
\draw (1,10.5) arc [radius = .5cm,start angle = 180, end angle = 0];
\draw (3,10.5) arc [radius = .5cm,start angle = 180, end angle = 0];
\draw (5,10.5) arc [radius = .5cm,start angle = 180, end angle = 0];
\end{tikzpicture}
\caption{A standard $4$-plat diagram with an even number of rows.}
\label{f:evenplat}
\end{figure}

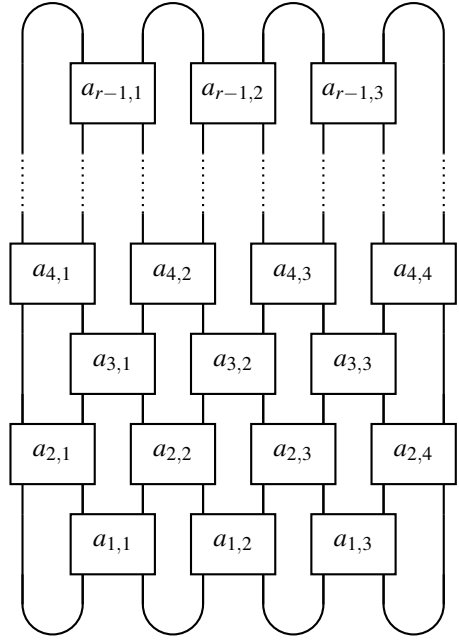
\begin{figure}
\begin{tikzpicture}[thick,scale=0.8]
\draw (0,0) -- (0,-.5) arc [radius = 0.5, start angle = 180, end angle = 360] -- (1,0);
\draw (2,0) -- (2,-.5) arc [radius = 0.5, start angle = 180, end angle = 360] -- (3,0);
\draw (4,0) -- (4,-.5) arc [radius = 0.5, start angle = 180, end angle = 360] -- (5,0);
\draw (6,0) -- (6,-.5) arc [radius = 0.5, start angle = 180, end angle = 360] -- (7,0);

\draw (.8,0) rectangle (2.2,1);
\draw (1.5,.5) node {$a_{1,1}$};
\draw (1,0) -- (1,-.5);
\draw (2,0) -- (2,-.5);
\draw (1,1) -- (1,1.5);
\draw (2,1) -- (2,1.5);

\begin{scope}[xshift = 2cm, yshift = 0]
\draw (.8,0) rectangle (2.2,1);
\draw (1.5,.5) node {$a_{1,2}$};
\draw (1,0) -- (1,-.5);
\draw (2,0) -- (2,-.5);
\draw (1,1) -- (1,1.5);
\draw (2,1) -- (2,1.5);
\end{scope}

\begin{scope}[xshift = 4cm, yshift = 0]
\draw (.8,0) rectangle (2.2,1);
\draw (1.5,.5) node {$a_{1,3}$};
\draw (1,0) -- (1,-.5);
\draw (2,0) -- (2,-.5);
\draw (1,1) -- (1,1.5);
\draw (2,1) -- (2,1.5);
\end{scope}

\begin{scope}[xshift = -1cm, yshift = 1.5cm]
\draw (.8,0) rectangle (2.2,1);
\draw (1.5,.5) node {$a_{2,1}$};
\draw (1,0) -- (1,-.5);
\draw (2,0) -- (2,-.5);
\draw (1,1) -- (1,1.5);
\draw (2,1) -- (2,1.5);
\end{scope}

\begin{scope}[xshift = 1cm, yshift = 1.5cm]
\draw (.8,0) rectangle (2.2,1);
\draw (1.5,.5) node {$a_{2,2}$};
\draw (1,0) -- (1,-.5);
\draw (2,0) -- (2,-.5);
\draw (1,1) -- (1,1.5);
\draw (2,1) -- (2,1.5);
\end{scope}

\begin{scope}[xshift = 3cm, yshift = 1.5cm]
\draw (.8,0) rectangle (2.2,1);
\draw (1.5,.5) node {$a_{2,3}$};
\draw (1,0) -- (1,-.5);
\draw (2,0) -- (2,-.5);
\draw (1,1) -- (1,1.5);
\draw (2,1) -- (2,1.5);
\end{scope}

\begin{scope}[xshift = 5cm, yshift = 1.5cm]
\draw (.8,0) rectangle (2.2,1);
\draw (1.5,.5) node {$a_{2,4}$};
\draw (1,0) -- (1,-.5);
\draw (2,0) -- (2,-.5);
\draw (1,1) -- (1,1.5);
\draw (2,1) -- (2,1.5);
\end{scope}

\draw (0,-.5) -- (0,1);
\draw (7,-.5) -- (7,1);

\begin{scope}[yshift = 3cm]

    \draw (.8,0) rectangle (2.2,1);
    \draw (1.5,.5) node {$a_{3,1}$};
    \draw (1,0) -- (1,-.5);
    \draw (2,0) -- (2,-.5);
    \draw (1,1) -- (1,1.5);
    \draw (2,1) -- (2,1.5);
    
    \begin{scope}[xshift = 2cm, yshift = 0]
    \draw (.8,0) rectangle (2.2,1);
    \draw (1.5,.5) node {$a_{3,2}$};
    \draw (1,0) -- (1,-.5);
    \draw (2,0) -- (2,-.5);
    \draw (1,1) -- (1,1.5);
    \draw (2,1) -- (2,1.5);
    \end{scope}
    
    \begin{scope}[xshift = 4cm, yshift = 0]
    \draw (.8,0) rectangle (2.2,1);
    \draw (1.5,.5) node {$a_{3,3}$};
    \draw (1,0) -- (1,-.5);
    \draw (2,0) -- (2,-.5);
    \draw (1,1) -- (1,1.5);
    \draw (2,1) -- (2,1.5);
    \end{scope}
    
    \begin{scope}[xshift = -1cm, yshift = 1.5cm]
    \draw (.8,0) rectangle (2.2,1);
    \draw (1.5,.5) node {$a_{4,1}$};
    \draw (1,0) -- (1,-.5);
    \draw (2,0) -- (2,-.5);
    \draw (1,1) -- (1,1.5);
    \draw (2,1) -- (2,1.5);
    \end{scope}
    
    \begin{scope}[xshift = 1cm, yshift = 1.5cm]
    \draw (.8,0) rectangle (2.2,1);
    \draw (1.5,.5) node {$a_{4,2}$};
    \draw (1,0) -- (1,-.5);
    \draw (2,0) -- (2,-.5);
    \draw (1,1) -- (1,1.5);
    \draw (2,1) -- (2,1.5);
    \end{scope}
    
    \begin{scope}[xshift = 3cm, yshift = 1.5cm]
    \draw (.8,0) rectangle (2.2,1);
    \draw (1.5,.5) node {$a_{4,3}$};
    \draw (1,0) -- (1,-.5);
    \draw (2,0) -- (2,-.5);
    \draw (1,1) -- (1,1.5);
    \draw (2,1) -- (2,1.5);
    \end{scope}
    
    \begin{scope}[xshift = 5cm, yshift = 1.5cm]
    \draw (.8,0) rectangle (2.2,1);
    \draw (1.5,.5) node {$a_{4,4}$};
    \draw (1,0) -- (1,-.5);
    \draw (2,0) -- (2,-.5);
    \draw (1,1) -- (1,1.5);
    \draw (2,1) -- (2,1.5);
    \end{scope}

    \draw[dotted] (0,3) -- (0,4);
    \draw[dotted] (1,3) -- (1,4);
    \draw[dotted] (2,3) -- (2,4);
    \draw[dotted] (3,3) -- (3,4);
    \draw[dotted] (4,3) -- (4,4);
    \draw[dotted] (5,3) -- (5,4);
    \draw[dotted] (6,3) -- (6,4);
    \draw[dotted] (7,3) -- (7,4);
    
    \draw (0,-.5) -- (0,1);
    \draw (7,-.5) -- (7,1);

\end{scope}

\begin{scope}[yshift = 7.5cm]

    \draw (.8,0) rectangle (2.2,1);
    \draw (1.5,.5) node {$a_{r-1,1}$};
    \draw (1,0) -- (1,-.5);
    \draw (2,0) -- (2,-.5);
    \draw (1,1) -- (1,1.5);
    \draw (2,1) -- (2,1.5);
    
    \begin{scope}[xshift = 2cm, yshift = 0]
    \draw (.8,0) rectangle (2.2,1);
    \draw (1.5,.5) node {$a_{r-1,2}$};
    \draw (1,0) -- (1,-.5);
    \draw (2,0) -- (2,-.5);
    \draw (1,1) -- (1,1.5);
    \draw (2,1) -- (2,1.5);
    \end{scope}
    
    \begin{scope}[xshift = 4cm, yshift = 0]
    \draw (.8,0) rectangle (2.2,1);
    \draw (1.5,.5) node {$a_{r-1,3}$};
    \draw (1,0) -- (1,-.5);
    \draw (2,0) -- (2,-.5);
    \draw (1,1) -- (1,1.5);
    \draw (2,1) -- (2,1.5);
    \end{scope}

    \draw (0,-.5) -- (0,1);
    \draw (7,-.5) -- (7,1);

\end{scope}

\begin{scope}[yshift=8.5cm,yscale=-1]
\draw (0,0) -- (0,-.5) arc [radius = 0.5, start angle = 180, end angle = 360] -- (1,0);
\draw (2,0) -- (2,-.5) arc [radius = 0.5, start angle = 180, end angle = 360] -- (3,0);
\draw (4,0) -- (4,-.5) arc [radius = 0.5, start angle = 180, end angle = 360] -- (5,0);
\draw (6,0) -- (6,-.5) arc [radius = 0.5, start angle = 180, end angle = 360] -- (7,0);
\end{scope}
\end{tikzpicture}
\caption{A standard $4$-plat diagram with an odd number of rows.}
\label{f:oddplat}
\end{figure}

Suppose $\Sigma$ is a bridge sphere for a knot $K$.  The \emph{curve graph} $\cC(\Sigma_K)$ of $\Sigma_K$ is the infinite simplicial graph defined as follows:
\begin{itemize}
\item the vertices of $\cC(\Sigma_K)$ are the isotopy classes $[\gamma]$ of simple closed curves $\gamma \subset \Sigma_K$ that are non-trivial and non-peripheral (meaning $\gamma$ bounds neither a disk nor a once-punctured disk in $\Sigma_K$), and
\item two vertices $v_1 \ne v_2$ of $\cC(\Sigma_K)$ are connected by an edge if and only if $v_1 = [\gamma_1]$ and $v_2 = [\gamma_2]$ where $\gamma_1 \cap \gamma_2 = \emptyset$.
\end{itemize}
If $v_1$ and $v_2$ are two vertices of $\cC(\Sigma_K)$, we let $d(v_1,v_2)$ denote their distance from one another with respect to the combinatorial path metric on $\cC(\Sigma_K)$.  More generally, if $A$ and $B$ are two subsets of the vertex set of $\cC(\Sigma_K)$, we define
\[ d(A,B) \defeq \min \{ d(a,b) \mid a \in A, b \in B \}. \]
Computing or bounding the distance between various subsets of curve graphs is a basic and important problem in combinatorial topology, as the following exemplifies.

Let $[\gamma]$ be a vertex of $\cC(\Sigma_K)$.  We say $[\gamma]$ \emph{bounds a disk} in $V_K$ if there exists a properly embedded disk $D \subset V_K$ with $\partial D = \gamma$.  Let $\mathcal{D}(V_K)$ denote the \emph{disk set} of $V_K$ consisting of all vertices in $\cC(\Sigma_K)$ that bound a disk in $V_K$.  Likewise, $\mathcal{D}(W_K)$ denotes the disk set of $W_K$.  The \emph{distance} of the bridge decomposition of $K$ by $\Sigma$ is
\[ d(K,\Sigma) \defeq d(\mathcal{D}(V_K),\mathcal{D}(W_K)). \]
We refer the reader to \cite{BachmanSchleimer:bridge,Tomova:bridge,JohnsonMoriah:plat} for much more information about bridge distance.\footnote{Note that the first reference uses a slightly different (but equivalent) definition of distance.  Our definition follows \cite{Tomova:bridge} and \cite{JohnsonMoriah:plat}.}  We summarize everything we need with the following mega-proposition.

\begin{prop}[After Johnson and Moriah {\cite[Thm.~1.2]{JohnsonMoriah:plat}}]
If $K$ is a highly twisted standard $m$-plat link diagram with $n$ rows where $m \ge 3$ and $n>4m(m-2)$ then the following hold:
\begin{enumerate}
\item the distance of the induced bridge sphere $\Sigma$ is $d(K,\Sigma) = \lceil n/(2(m-2)) \rceil$,
\item the bridge number $m$ is minimal (over all diagrams of equivalent links) and the induced bridge sphere of $K$ is the unique minimal bridge sphere (up to isotopy), and
\item $K$ is a hyperbolic link.
\end{enumerate}
\label{p:twisted}
\end{prop}

\begin{proof}
\begin{enumerate}
\item This is the conclusion of \cite[Thm.~1.2]{JohnsonMoriah:plat}, which holds without the assumption $n > 4m(m-2)$).
\item The assumptions on $m$ and $n$ imply $d(K,\Sigma)>2m$, so Tomova's work \cite{Tomova:bridge} shows $\Sigma$ is the unique minimal bridge sphere of $K$.
\item Since $d(\Sigma)>2$, $K$ is prime, atoroidal, and an-annular.  Thus $K$ is hyperbolic by Thurston's hyperbolization theorem. (The same argument is found in the proof of \cite[Cor.~6.2]{BachmanSchleimer:bridge}, for example.)
\end{enumerate}
\vspace{-\baselineskip}
\end{proof}

In fact, if $K$ is a highly twisted knot diagram satisfying the conditions of Proposition \ref{p:twisted} that also happens to be alternating, then Agol and Thurston's refinement of a theorem of Lackenby \cite{Lackenby:alternating} shows the hyperbolic volume of $K$ satisfies
\[ v_3(t(K) - 2) \le \mathrm{Vol}(S^3 \setminus K) \le 10 v_3(t(K) - 1) \]
where $v_3 \approx 1.01494$ is the volume of a regular hyperbolic ideal 3-simple and
\[ t(K) = \lfloor(2m-1)(n/2)\rfloor\]
is the number of twist regions of $K$.  We do not use or need this result.  Rather, we include it simply to provide the reader some sense of how the hyperbolic geometry of such a knot depends on its diagram.  The main behavior to note is that so long as the diagram is highly twisted, the precise twist coefficients don't affect the volume very much, as the volume is coarsely proportional to the number of twist regions.  In other words, $\mathrm{Vol}(S^3 \setminus K) = \Theta(t(K))$.

\section{Proofs}
\label{s:proof}
The proofs of Theorems \ref{thm:reduction1} and \ref{thm:reduction2} use the same basic idea.  Given a knot diagram $K$, we pad it by inserting additional crossings in the diagram.  The result is another diagram $K'$ of a different knot, but if we pad with an appropriate number of crossings in a row (see Vafa's theorem, Lemma \ref{l:vafa} below), then we can guarantee $Z(K) = Z(K')$.  If we perform this kind of padding in enough places, and also use some other tricks, then we can apply the results of Section \ref{s:diagrams} to build $K'$ and $K''$ with the desired properties in the conclusions of the two theorems.  

\subsection{Braiding exponents in modular fusion categories}
\label{ss:vafa}
Given a pair of objects $V,W$ in a modular fusion category $\cC$, define the \emph{braiding exponent} $e(V,W)$ to be the order of the ``squared braiding" morphism $br_{V,W}^2 \defeq br_{W,V} \circ br_{V,W}$.  The \emph{self-braiding exponent} of a single object $V$ is $e(V,V)$.  Vafa's theorem says $e(V,W)$ is a well-defined integer.

\begin{lem}[Vafa's theorem \cite{Vafa:conformal}, see also \cite{AndersonMoore:rationality,Etingof:Vafa}]
In any $(2+1)$-dimensional Witten-Reshetikin-Turaev TQFT determined by a modular fusion category, the square of the braiding of any two objects is a \emph{finite order} linear map.
\label{l:vafa}
\end{lem}

\subsection{Proof of Theorem \ref{thm:reduction2}}
\label{ss:reduction2}
Theorem \ref{thm:reduction2} follows immediately by combining the following result with Proposition \ref{p:twisted}.

\begin{prop}
Fix a modular fusion category $\cC$ and an object $V$ with a scalar twist $\theta$.  Let $Z$ be the Witten-Reshetikhin-Turaev invariant of oriented ribbon knots colored by $V$.  Then there exists a classical polynomial time algorithm that converts an oriented ribbon knot diagram $K$ into an alternating and highly twisted standard $m$-plat ribbon diagram $K''$ with $Z(K'') = Z(K)$ and $n=2k$ rows, where $m \ge 3$ and $n > 4m(m-2)$.
\label{p:reduction}
\end{prop}

\begin{proof}
Using Vafa's theorem, we fix an even constant $2T\ge e(V,V)>0$ (depending on $\cC$ and $V$) such that adding a string of $2T$ crossings (all with the same sign) anywhere in a diagram produces a new knot diagram with the same $Z$ invariant.  That is, we pick $T$ so that
\[
\begin{tikzpicture}[scale=0.18]
\draw (-3,0) node[left] {\large $Z($};
\draw[thick] plot[smooth, tension=.7] coordinates {(-3,1) (0,.7) (3,1) };
\draw[thick] plot[smooth, tension=.7] coordinates {(-3,-1) (0,-.7) (3,-1) };
\draw (3,0) node[right] {\large $)$};
\draw (6,0) node {$=$};
\begin{scope}[xshift= 13.5cm]
\draw (-3,0) node[left] {\large$Z($};
\draw[thick] plot[smooth, tension=.7] coordinates {(-3,1) (0,.7) (3,1) };
\draw[thick] plot[smooth, tension=.7] coordinates {(-3,-1) (0,-.7) (3,-1) };
\draw[thick,fill=white] (-2.2,-1.5) rectangle (2.2,1.2);
\draw (0,0) node {$\pm 2T$};
\draw (3,0) node[right] {\large$).$};
\end{scope}
\end{tikzpicture}
\]

Let $K = K_0$ be our initial (oriented, ribbon) knot diagram.  We reduce it to $K''$ in two steps.

First, if necessary, we apply a sequence of Reidemeister 1 and 2 moves to $K_0$ to get a regularly isotopic diagram $K_1$ that is a standard $m$-plat diagram where $m \ge 3$, and the number of rows $n$ is greater than $4m(m-2)$.  In detail, let $m$ be the number of local maxima in $K_0$ (which equals the number of local minima) with respect to some preferred coordinates on the diagram plane.  If $m < 3$, then introduce a pair of Reidemeister 1 moves whose framings cancel
(recall from Section \ref{ss:notation} that we use the blackboard framing).  The result will have $m \ge 3$.  Now apply a sequence of Reidemeister 2 moves to make a standard $m$-plat diagram.  Since we are not requiring $K_1$ be highly twisted, if necessary, we introduce twist regions with coefficients $a_{i,j}=0$ to ensure that $n > 4m(m-2)$.  Since $K_1$ is regularly isotopic to $K_0$, $Z(K_1) = Z(K_0)$.  This step takes quadratic time in the crossing number of $K_0$ and does not depend on $\cC$ or $V$.

Second, we increase the twist coefficients of $K_1$ in order to get a highly twisted alternating diagram.  Let the twist coefficients of $K_1$ be $a_{i,j}$.  Build $K_2$ by replacing each $a_{i,j}$ in $K_1$ with $a_{i,j}'$ as follows:
\[ a_{i,j}' = \begin{cases} a_{i,j} + 2T & \text{ if $a_{i,j} \ge 0$,} \\ a_{i,j} - 2T & \text{ otherwise.}\end{cases}\]
This step takes a linear amount of time in the crossing number of $K_1$ (note that if $a_{i,j}=0$, we still count it as a single ``crossing").

Let $K''=K_2$.  It is evident from the construction that $K''$ has the same number of plats and rows as $K_1$, and the assumption $2T > 2$ guarantees $|a'_{i,j}| \ge 3$ for all $i,j$.  Therefore, $K''$ is a highly-twisted standard $m$-plat with $n$ rows where $m \ge 3$ and $n > 4m(m-2)$.  Of course, by our choice of $T$,
\[ Z(K'') = Z(K_2) = Z(K_1) = Z(K_0) = Z(K).\]
The entire procedure taking $K$ to $K''$ takes quadratic time in the crossing number of $K$ and is linear in $T$.  Thus if we also allow $\cC$ and $V$ to vary, then the reduction runs in time that is jointly polynomial in the self-braiding exponent $e(V,V)$ and the crossing number of $K$.
\end{proof}

\subsection{Proof of Theorem \ref{thm:reduction1}}
\label{ss:reduction1}

As in the previous subsection, fix an even constant $2T\ge e(V,V)>2$ (depending on $\cC$ and $V$) such that adding $2T$ consecutive crossings anywhere in the diagram produces a new knot diagram with the same $Z$ invariant.

Let $K$ be our initial diagram.  We must describe a polynomial time algorithm that builds a diagram $K'$ that is alternating, has minimal crossing number, is hyperbolic, and satisfies
\[ Z(K') = \theta^{r(K)} Z(K)\]
for some polynomial time computable function $r(K)$.  We perform the reduction in three steps, with a conditional fourth step.  Let $K_0=K$. 

First, we alter some of the crossings of $K_0$ in order to get an alternating diagram $K_1$ with $Z(K_1) = Z(K_0)$.  This is performed as follows.  If $K_0$ is not alternating, then in linear time we may identify a nonempty set of crossings $v_1, v_2, \dots, v_k$ such that reversing them yields an alternating diagram.  However, doing this reversal will not necessarily preserve $Z$ (unless $T=1$).  Instead, for each of the positive crossings that need to be made negative, we replace them with $2T-1$ negative crossings, and vice versa for negative crossings that need to be made positive.  The resulting diagram $K_2$ is an alternating knot diagram, and each of these crossing replacements can be understood in two steps: insert $2T$ additional twists in the opposite direction, then apply a Reidemeister 2 move.  Thus $Z(K_1) = Z(K_0)$.  This entire padding process takes linear time both in the crossing number of $K_1$ and in $T$.

Second, we use Lemma \ref{l:reduced} to reduce $K_1$ to a minimal crossing alternating diagram $K_2$ in polynomial time.  The removal of each nugatory crossing in the reduction algorithm changes the writhe of $K_1$ by $\pm 1$.  If we let $r'(K)$ be the total change in the writhe $w(K_2)-w(K_1)$, then
\[ Z(K_2) = \theta^{w(K_2)-w(K_1)}Z(K_1) = \theta^{r'(K)}Z(K).\]
Clearly $r'(K)$ is computable in polynomial time from $K$.  Except for the value of $\theta$, this step does not depend on $\cC$ and $V$.

Third, we use Lemma \ref{l:prime} to identify the diagrammatic connected sum decomposition $K_2 = \#_{i=1}^k P_i$ and then perform some additional twist padding to create a new diagram $K_3$ from $K_2$ that is reduced, prime, and alternating.  In detail, we build $K_3$ as follows.  If $k=0$ or 1 (that is, if $K_2$ is trivial or prime already), then simply set $K_3=K_2$.  If $k>1$, then by uniqueness of connected sum decompositions we may assume without loss of generality that $K_2$ looks like this:
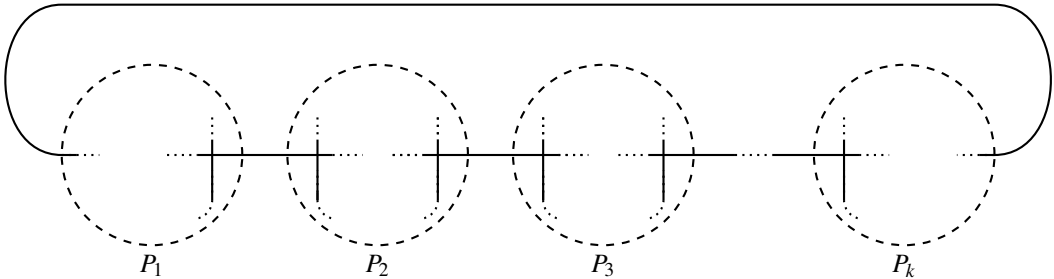
\[
\begin{tikzpicture}[thick]
\draw[dashed] (1,0) circle [radius=1.2];
\draw (1,-1.2) node[below] {$P_1$};

\draw (1.6,0) -- (3.4,0);
\draw (1.8,-.6) -- (1.8,.2);
\draw (3.2,-.6) -- (3.2,.2);

\draw[dotted] (1.2,0) -- (1.6,0);
\draw[dotted] (3.4,0) -- (3.8,0);
\draw[dotted] (1.8,-.5) -- (1.8,-.2);
\draw[dotted] (1.8,.5) -- (1.8,.2);
\draw[dotted] (3.2,-.5) -- (3.2,-.2);
\draw[dotted] (3.2,.5) -- (3.2,.2);

\draw[dotted] (1.8,-.6) arc [start angle = 0, end angle = -90, radius=.25];
\draw (3.2,-.5) -- (3.2,-.6);
\draw[dotted] (3.2,-.6) arc [start angle = 180, end angle = 270, radius=.25];

\draw[dashed] (4,0) circle [radius=1.2];
\draw (4,-1.2) node[below] {$P_2$};

\begin{scope}[xshift=3cm]
\draw (1.6,0) -- (3.4,0);
\draw (1.8,-.6) -- (1.8,.2);
\draw (3.2,-.6) -- (3.2,.2);
\draw[dotted] (1.2,0) -- (1.6,0);
\draw[dotted] (3.4,0) -- (3.8,0);
\draw[dotted] (1.8,-.5) -- (1.8,-.2);
\draw[dotted] (1.8,.5) -- (1.8,.2);
\draw[dotted] (3.2,-.5) -- (3.2,-.2);
\draw[dotted] (3.2,.5) -- (3.2,.2);

\draw[dotted] (1.2,0) -- (1.6,0);
\draw[dotted] (3.4,0) -- (3.8,0);
\draw[dotted] (1.8,-.5) -- (1.8,-.2);
\draw[dotted] (1.8,.5) -- (1.8,.2);
\draw[dotted] (3.2,-.5) -- (3.2,-.2);
\draw[dotted] (3.2,.5) -- (3.2,.2);

\draw[dotted] (1.8,-.6) arc [start angle = 0, end angle = -90, radius=.25];
\draw (3.2,-.5) -- (3.2,-.6);
\draw[dotted] (3.2,-.6) arc [start angle = 180, end angle = 270, radius=.25];
\end{scope}

\draw[dashed] (7,0) circle [radius=1.2];
\draw (7,-1.2) node[below] {$P_3$};

\begin{scope}[xshift=6cm]
\draw (1.6,0) -- (2.75,0);
\draw[dotted] (2.75,0) -- (3.25,0);
\draw (1.8,-.6) -- (1.8,.2);
\draw[dotted] (1.2,0) -- (1.6,0);
\draw[dotted] (1.8,-.5) -- (1.8,-.2);
\draw[dotted] (1.8,.5) -- (1.8,.2);

\draw[dotted] (1.8,-.6) arc [start angle = 0, end angle = -90, radius=.25];

\end{scope}

\begin{scope}[xshift=7cm]
\draw (2.25,0) -- (3.4,0);
\draw (3.2,-.6) -- (3.2,.2);
\draw[dotted] (3.4,0) -- (3.8,0);
\draw[dotted] (3.2,-.5) -- (3.2,-.2);
\draw[dotted] (3.2,.5) -- (3.2,.2);

\draw[dotted] (3.2,-.6) arc [start angle = 180, end angle = 270, radius=.25];

\end{scope}

\draw[dashed] (11,0) circle [radius=1.2];
\draw (11,-1.2) node[below] {$P_k$};

\draw[dotted] (-.2,0) -- (.3,0);
\draw (0,0) -- (-.2,0) .. controls (-1.2,0) and (-1.2,2) .. (-.2,2) -- (12.2,2) .. controls (13.2,2) and (13.2,0) .. (12.2,0) -- (12,0);
\draw[dotted] (12.2,0) -- (11.7,0);
\end{tikzpicture}
\]
where we've suppressed the crossing data for convenience.  We pad this diagram with additional crossings as follows and call the result $K_3$:
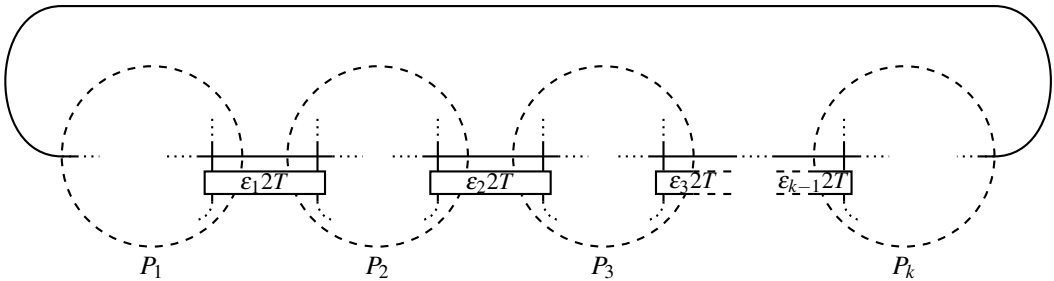
\[
\begin{tikzpicture}[thick]

\draw[dashed] (1,0) circle [radius=1.2];
\draw (1,-1.2) node[below] {$P_1$};

\draw (1.6,0) -- (3.4,0);
\draw (1.8,-.2) -- (1.8,.2);
\draw (3.2,-.2) -- (3.2,.2);

\draw[dotted] (1.2,0) -- (1.6,0);
\draw[dotted] (3.4,0) -- (3.8,0);
\draw[dotted] (1.8,-.5) -- (1.8,-.2);
\draw[dotted] (1.8,.5) -- (1.8,.2);
\draw[dotted] (3.2,-.5) -- (3.2,-.2);
\draw[dotted] (3.2,.5) -- (3.2,.2);

\draw[dashed] (4,0) circle [radius=1.2];
\draw (4,-1.2) node[below] {$P_2$};

\begin{scope}[xshift=3cm]
\draw (1.6,0) -- (3.4,0);
\draw (1.8,-.2) -- (1.8,.2);
\draw (3.2,-.2) -- (3.2,.2);
\draw[dotted] (1.2,0) -- (1.6,0);
\draw[dotted] (3.4,0) -- (3.8,0);
\draw[dotted] (1.8,-.5) -- (1.8,-.2);
\draw[dotted] (1.8,.5) -- (1.8,.2);
\draw[dotted] (3.2,-.5) -- (3.2,-.2);
\draw[dotted] (3.2,.5) -- (3.2,.2);
\end{scope}

\draw[dashed] (7,0) circle [radius=1.2];
\draw (7,-1.2) node[below] {$P_3$};

\begin{scope}[xshift=6cm]
\draw (1.6,0) -- (2.75,0);
\draw[dotted] (2.75,0) -- (3.25,0);
\draw (1.8,-.2) -- (1.8,.2);
\draw[dotted] (1.2,0) -- (1.6,0);
\draw[dotted] (1.8,-.5) -- (1.8,-.2);
\draw[dotted] (1.8,.5) -- (1.8,.2);
\end{scope}

\begin{scope}[xshift=7cm]
\draw (2.25,0) -- (3.4,0);
\draw (3.2,-.2) -- (3.2,.2);
\draw[dotted] (3.4,0) -- (3.8,0);
\draw[dotted] (3.2,-.5) -- (3.2,-.2);
\draw[dotted] (3.2,.5) -- (3.2,.2);
\end{scope}

\draw[dashed] (11,0) circle [radius=1.2];
\draw (11,-1.2) node[below] {$P_k$};

\draw[dotted] (-.2,0) -- (.3,0);
\draw (0,0) -- (-.2,0) .. controls (-1.2,0) and (-1.2,2) .. (-.2,2) -- (12.2,2) .. controls (13.2,2) and (13.2,0) .. (12.2,0) -- (12,0);
\draw[dotted] (12.2,0) -- (11.7,0);

\draw[fill=white] (1.7,-.5) rectangle (3.3,-.2);
\draw (2.5,-.35) node{\small$\epsilon_1 2T$};
\draw (1.8,-.5) -- (1.8,-.6);
\draw[dotted] (1.8,-.6) arc [start angle = 0, end angle = -90, radius=.25];
\draw (3.2,-.5) -- (3.2,-.6);
\draw[dotted] (3.2,-.6) arc [start angle = 180, end angle = 270, radius=.25];

\begin{scope}[xshift=3cm]
\draw[fill=white] (1.7,-.5) rectangle (3.3,-.2);
\draw (2.5,-.35) node{\small$\epsilon_2 2T$};
\draw (1.8,-.5) -- (1.8,-.6);
\draw[dotted] (1.8,-.6) arc [start angle = 0, end angle = -90, radius=.25];
\draw (3.2,-.5) -- (3.2,-.6);
\draw[dotted] (3.2,-.6) arc [start angle = 180, end angle = 270, radius=.25];
\end{scope}

\begin{scope}[xshift=6cm]
\draw[draw=white,fill=white] (1.7,-.5) rectangle (2.8,-.2);
\draw (2.2,-.2) -- (1.7,-.2) -- (1.7,-.5) -- (2.2,-.5);
\draw[dashed] (2.7,-.2) -- (2.2,-.2);
\draw[dashed] (2.7,-.5) -- (2.2,-.5);
\draw (2.2,-.35) node{\small$\epsilon_3 2T$};
\draw (1.8,-.5) -- (1.8,-.6);
\draw[dotted] (1.8,-.6) arc [start angle = 0, end angle = -90, radius=.25];
\end{scope}

\begin{scope}[xshift=7cm]
\draw[draw=white,fill=white] (2.2,-.5) rectangle (3.3,-.2);
\draw[dashed] (2.3,-.2) -- (2.8,-.2);
\draw[dashed] (2.3,-.5) -- (2.8,-.5);
\draw (2.8,-.2) -- (3.3,-.2) -- (3.3,-.5) -- (2.8,-.5);
\draw (2.8,-.35) node{\small$\epsilon_{k-1} 2T$};
\draw (3.2,-.5) -- (3.2,-.6);
\draw[dotted] (3.2,-.6) arc [start angle = 180, end angle = 270, radius=.25];
\end{scope}

\end{tikzpicture}
\]
Here in diagram $K_3$ each $\epsilon_i$ equals either $+1$ or $-1$; we take whichever sign is necessary to ensure $K_3$ is alternating.  This choice depends in a simple way on the two crossings involved in connected-summing $P_i$ and $P_{i+1}$, with two possible cases.  Namely, if the arc leaving $P_i$ enters $P_{i+1}$ as an under-crossing (hence, left $P_i$ from an over-crossing), then $\epsilon_i = -1$; otherwise $\epsilon_i= +1$.  These $\epsilon_i$'s can be determined in polynomial time, and hence $K_3$ can be built from $K_2$ in polynomial time.  Clearly $Z(K_3) = Z(K_2)$, and it is straightforward to check that $K_3$ is reduced and prime.  Note that this step is linear in $T$.

Finally, by Lemma \ref{l:menasco} $K_3$ is either hyperbolic, the unique minimal crossing alternating diagram of a $(2,p)$ torus knot for some odd $p$, or the trivial unknot diagram, and we can recognize which of these is the case in time polynomial in the crossing number of $K_3$.

If $K_3$ is hyperbolic (this is by far the likeliest outcome), then we let $K'=K_3$ and $r(K) = r'(K)$.

If $K_3$ is the minimal crossing alternating $(2,p)$ torus knot diagram, then we use the Vafa padding trick and Proposition \ref{p:twisted} as in Section \ref{ss:reduction2} to build a highly twisted alternating knot $K_4$ that has no nugatory crossings, satisfies the conditions of Proposition \ref{p:twisted} necessary to guarantee it is hyperbolic, and has $Z(K_4) = \theta \cdot Z(K_3)$.  Indeed, if $p>0$, pick any one of the diagrams as in Figure \ref{f:torus} with at least $n>4\times 3 (3-2) = 12$ rows.  If $p<0$, do the same, but switch all of the signs in the picture.  This reduction works in linear time in $p$.  Now let $K'=K_4$ and $r(K) = r'(K) + 1$.

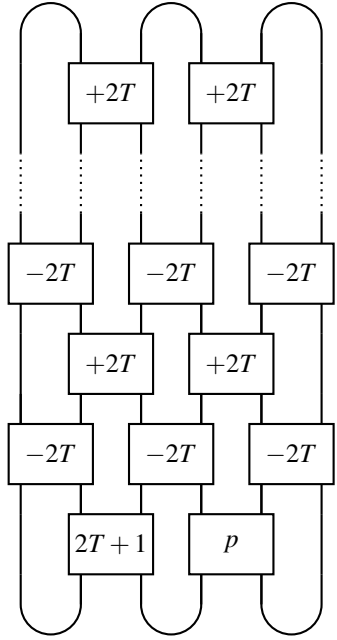
\begin{figure}
\begin{tikzpicture}[thick,scale=0.8]
\draw (0,0) -- (0,-.5) arc [radius = 0.5, start angle = 180, end angle = 360] -- (1,0);
\draw (2,0) -- (2,-.5) arc [radius = 0.5, start angle = 180, end angle = 360] -- (3,0);
\draw (4,0) -- (4,-.5) arc [radius = 0.5, start angle = 180, end angle = 360] -- (5,0);

\draw (0,0) -- (0,1);
\draw (5,0) -- (5,1);

\draw (.8,0) rectangle (2.2,1);
\draw (1.5,.5) node {$2T+1$};
\draw (1,0) -- (1,-.5);
\draw (2,0) -- (2,-.5);
\draw (1,1) -- (1,1.5);
\draw (2,1) -- (2,1.5);

\begin{scope}[xshift = 2cm, yshift = 0]
\draw (.8,0) rectangle (2.2,1);
\draw (1.5,.5) node {$p$};
\draw (1,0) -- (1,-.5);
\draw (2,0) -- (2,-.5);
\draw (1,1) -- (1,1.5);
\draw (2,1) -- (2,1.5);
\end{scope}

\begin{scope}[xshift = -1cm, yshift = 1.5cm]
\draw (.8,0) rectangle (2.2,1);
\draw (1.5,.5) node {$-2T$};
\draw (1,0) -- (1,-.5);
\draw (2,0) -- (2,-.5);
\draw (1,1) -- (1,1.5);
\draw (2,1) -- (2,1.5);
\end{scope}

\begin{scope}[xshift = 1cm, yshift = 1.5cm]
\draw (.8,0) rectangle (2.2,1);
\draw (1.5,.5) node {$-2T$};
\draw (1,0) -- (1,-.5);
\draw (2,0) -- (2,-.5);
\draw (1,1) -- (1,1.5);
\draw (2,1) -- (2,1.5);
\end{scope}

\begin{scope}[xshift = 3cm, yshift = 1.5cm]
\draw (.8,0) rectangle (2.2,1);
\draw (1.5,.5) node {$-2T$};
\draw (1,0) -- (1,-.5);
\draw (2,0) -- (2,-.5);
\draw (1,1) -- (1,1.5);
\draw (2,1) -- (2,1.5);
\end{scope}

\begin{scope}[yshift = 3cm]

    \draw (.8,0) rectangle (2.2,1);
    \draw (1.5,.5) node {$+2T$};
    \draw (1,0) -- (1,-.5);
    \draw (2,0) -- (2,-.5);
    \draw (1,1) -- (1,1.5);
    \draw (2,1) -- (2,1.5);
    
    \begin{scope}[xshift = 2cm, yshift = 0]
    \draw (.8,0) rectangle (2.2,1);
    \draw (1.5,.5) node {$+2T$};
    \draw (1,0) -- (1,-.5);
    \draw (2,0) -- (2,-.5);
    \draw (1,1) -- (1,1.5);
    \draw (2,1) -- (2,1.5);
    \end{scope}
    
    \begin{scope}[xshift = -1cm, yshift = 1.5cm]
    \draw (.8,0) rectangle (2.2,1);
    \draw (1.5,.5) node {$-2T$};
    \draw (1,0) -- (1,-.5);
    \draw (2,0) -- (2,-.5);
    \draw (1,1) -- (1,1.5);
    \draw (2,1) -- (2,1.5);
    \end{scope}
    
    \begin{scope}[xshift = 1cm, yshift = 1.5cm]
    \draw (.8,0) rectangle (2.2,1);
    \draw (1.5,.5) node {$-2T$};
    \draw (1,0) -- (1,-.5);
    \draw (2,0) -- (2,-.5);
    \draw (1,1) -- (1,1.5);
    \draw (2,1) -- (2,1.5);
    \end{scope}
    
    \begin{scope}[xshift = 3cm, yshift = 1.5cm]
    \draw (.8,0) rectangle (2.2,1);
    \draw (1.5,.5) node {$-2T$};
    \draw (1,0) -- (1,-.5);
    \draw (2,0) -- (2,-.5);
    \draw (1,1) -- (1,1.5);
    \draw (2,1) -- (2,1.5);
    \end{scope}
    
    \draw[dotted] (0,3) -- (0,4);
    \draw[dotted] (1,3) -- (1,4);
    \draw[dotted] (2,3) -- (2,4);
    \draw[dotted] (3,3) -- (3,4);
    \draw[dotted] (4,3) -- (4,4);
    \draw[dotted] (5,3) -- (5,4);
    
    \draw (0,-.5) -- (0,1);
    \draw (5,0) -- (5,1);

\end{scope}

\begin{scope}[yshift = 7.5cm]

    \draw (.8,0) rectangle (2.2,1);
    \draw (1.5,.5) node {$+2T$};
    \draw (1,0) -- (1,-.5);
    \draw (2,0) -- (2,-.5);
    \draw (1,1) -- (1,1.5);
    \draw (2,1) -- (2,1.5);
    
    \begin{scope}[xshift = 2cm, yshift = 0]
    \draw (.8,0) rectangle (2.2,1);
    \draw (1.5,.5) node {$+2T$};
    \draw (1,0) -- (1,-.5);
    \draw (2,0) -- (2,-.5);
    \draw (1,1) -- (1,1.5);
    \draw (2,1) -- (2,1.5);
    \end{scope}
    
    \draw (0,-.5) -- (0,1);
	\draw (5,-.5) -- (5,1);

\end{scope}

\begin{scope}[yshift=8.5cm,yscale=-1]
\draw (0,0) -- (0,-.5) arc [radius = 0.5, start angle = 180, end angle = 360] -- (1,0);
\draw (2,0) -- (2,-.5) arc [radius = 0.5, start angle = 180, end angle = 360] -- (3,0);
\draw (4,0) -- (4,-.5) arc [radius = 0.5, start angle = 180, end angle = 360] -- (5,0);
\end{scope}
\end{tikzpicture}

\caption{Padding a $(2,p)$ torus knot to make it alternating and hyperbolic.}
\label{f:torus}
\end{figure}

If $K_3$ is the trivial unknot diagram, then by similar reasoning setting $p=2T+1$ in Figure \ref{f:torus} and choosing such a diagram with $n>12$ rows yields an alternating hyperbolic knot diagram $K_4$ with no nugatory crossings and $Z(K_4) = \theta^{2}Z(K_3)$ for some polynomial time $r''(K)$.  Let $K' = K_4$ and $r(K) = r''(K) + 2$.  (We note that there are arguably ``better" ways to cover this case.  Our approach here has the benefit of being able to reuse Figure \ref{f:torus}.)

In either of these three cases, the dependence on $T$ is linear.  Thus, if we allow $\cC$ and $V$ to vary, the reduction works in time polynomial in $e(V,V)$.
\qed

\subsection{Proof of Corollary \ref{c:main}}
\label{ss:cmain}
Let $Z(-)$ be any one of the invariants considered in Corollary \ref{c:main}.  So, for example, we could have $Z(-) = V(-,e^{2 i \pi/7})$ or $Z(-) = \#H(-,A_5,(1\ 2\ 3\ 4\ 5)^{A_5})$ where $(1\ 2\ 3\ 4\ 5)^{A_5}$ is the conjugacy class of $(1 \ 2 \ 3\ 4 \ 5)$ in $A_5$.  Note that, at least in this subsection, we fix \emph{one} such $Z$, we do not work uniformly in all of them.

Recall that for a fixed root of unity $q$, the value of the Jones polynomial $V(K,q)$ can be identified (after some minor normalization) with the Witten-Reshetikhin-Turaev invariant of $K$ determined from the modular fusion category $U_q(\mathfrak{sl}_2)$-mod by coloring $K$ with the fundamental representation $V$ of the quantum group $U_q(\mathfrak{sl}_2)$.  Likewise, $\#H(K,G,C)$ can be computed as a WRT invariant using the Dijkgraaf-Witten TQFT based on the modular fusion category $DG$-mod, where $DG$ is the untwisted Drinfeld double of $G$.   Thus, Theorems \ref{thm:reduction1} and \ref{thm:reduction2} apply to these invariants.

In caricature, each of the main results of \cite{FreedmanLarsenWang:two}, \cite{Kuperberg:Jones} and \cite{me:coloring} is proved using a similar strategy: given $Z$, one first finds some model of reversible circuits and associated problem $L$ that is hard.  Then given a circuit $C$ in this model, one constructs a polynomial time reduction to a knot diagram $K_C$ such that computing (or approximating) $|Z(K_C)|$ allows one to compute $L(C)$.  Of course, there is significant work involved in proving that these reductions are possible, but this is all we need to know, since it means we can prove Corollary \ref{c:main} if we can accomplish the following: 
\begin{quote}
\emph{Fix $Z$ and fix $i=1$ or $2$.  Provide a classical polynomial time algorithm that converts a given knot diagram $K$ to another diagram $K'$ (possibly of a different knot!) so that $|Z(K')| = |Z(K)|$ and $K'$ satisfies all of the desired promises in Theorem $i$.}
\end{quote}
Indeed, if we can do this, then we can first reduce a circuit $C$ to the knot diagram $K_C$ exactly as in \cite{FreedmanLarsenWang:two}, \cite{Kuperberg:Jones} or \cite{me:coloring} (as appropriate, depending on $Z$), and then simply reduce $K_C$ to $(K_C)'$.  Theorems \ref{thm:reduction1} and \ref{thm:reduction2} say this is possible. \qed

\subsection{Proof of Corollary \ref{c:uniform}}
\label{ss:uniformity}
Fix a polynomial $p(n)$ as in the statement of Corollary \ref{c:uniform}.  Then the main result of \cite{AharonovArad:uniformity} implies that the problem of additively approximating $V(L,t)$ is $\BQP$-hard when $L$ is a \emph{link} diagram and $t=e^{2\pi i/p(n)}$.  Using the same observation that Kuperberg made in \cite{Kuperberg:Jones} in order to improve the Freedman-Larsen-Wang hardness results from \emph{links} to \emph{knots} (namely, that braid density implies \emph{pure} braid density, since the pure braid subgroup $PB_n$ is finite index in $B_n$), we may improve the result of \cite{AharonovArad:uniformity} from links to knots.

In caricature, this hardness result is proved by converting a quantum circuit $C$ on $m$ gates into a knot diagram $K_C$ with $n=\poly(m)$ crossings and $t=e^{2\pi i/p(n)}$ such that an ``additive approximation" of $V(K_C,t)$ approximates the acceptance probability of $C$.  As in the previous subsection, $V(K,t)$ can be identified (up to some minor normalization) with the Witten-Reshetikhin-Turaev invariant of $K$ determined from the modular fusion category $U_q(\mathfrak{sl}_2)$-mod by coloring $K$ with the fundamental representation $V$ of the quantum group $U_q(\mathfrak{sl}_2)$, except notice that now $q = t = e^{2\pi i/p(n)}$ depends on $n$.  

For $V$ the defining representation in the modular fusion category $U_q(\mathfrak{sl}_2)$-mod and $q=e^{2\pi i/k}$, it turns out that $e(V,V) = O(\poly(k))$.  Applying the uniformity statements in Theorems \ref{thm:reduction1} and \ref{thm:reduction2} completes the proof. \qed

\end{document}